\documentclass[twocolumn,10pt]{IEEEtran}
\usepackage{ifpdf, flushend}

\ifCLASSINFOpdf
  \usepackage[pdftex]{graphicx}
  \graphicspath{{../pdf/}{../jpeg/}}
  \DeclareGraphicsExtensions{.pdf,.jpeg,.png}
\else
  \usepackage[dvips]{graphicx}
  \graphicspath{{../eps/}}
  \DeclareGraphicsExtensions{.eps}
\fi
\usepackage[cmex10]{amsmath}
\usepackage {amssymb}
\usepackage{algorithmic}
\usepackage{array}
\usepackage{mdwmath}
\usepackage{mdwtab}
\usepackage{eqparbox}
\usepackage{url}
\usepackage{hyperref}
\usepackage{algorithm}
\usepackage{algorithmic}
\usepackage{epstopdf,cite,color}
\usepackage{amsthm} 

\newtheorem{lemma}{Lemma}

\renewcommand{\baselinestretch}{1}

\newcommand{\beq}{\begin{equation}}
\newcommand{\eeq}{\end{equation}}
\newcommand{\beqn}{\begin{eqnarray}}
\newcommand{\eeqn}{\end{eqnarray}}
\newcommand{\beqno}{\begin{eqnarray*}}
\newcommand{\eeqno}{\end{eqnarray*}}
\newcommand{\bma}{\begin{displaymath}}
\newcommand{\ema}{\end{displaymath}}
\newcommand{\bnu}{\begin{enumerate}}
\newcommand{\enu}{\end{enumerate}}
\newcommand{\bce}{\begin{center}}
\newcommand{\ece}{\end{center}}
\newcommand{\btb}{\begin{tabular}}
\newcommand{\etb}{\end{tabular}}

\begin{document}

\title{DISCO-Dynamic Interference Suppression for Radar and Communication Cohabitation}
\author{Tan~Le,~\IEEEmembership{Member,~IEEE,} and Van~Le
\thanks{T.~Le is with the School of Engineering, Architecture and Aviation, Hampton University, Hampton, VA 23669, USA. 
Email: tan.le@hamptonu.edu.\\
V.~Le is with the Advanced Technology Center, Virginia Beach, VA 23453, USA. 
Email: lev13408@gmail.com.}}

\maketitle

\begin{abstract}
We propose the joint dynamic power allocation and multi-relay selection for the cohabitation of high-priority military radar and low-priority commercial 5G communication. 
To improve the 5G network performance, we design the full-duplex underlay cognitive radio network for the low-priority commercial 5G network, where multiple relays are selected for concurrently receive the signal from the source and send it to the destination. 
Then, we propose the interference suppression at the high-priority radar system by using both non-coherent and coherent relay cases. 
In particular, we formulate the optimization problem for maximizing the system rate, with the consideration of the power constraints at the 5G users and the interference constraint at the radar system. 
Then, we propose the mathematical analysis model to evaluate the rate performance, considering the impacts of self-interference at the relays and derive the algorithms of joint power allocation and relay selection. 
Our numerical results demonstrate the characteristic of the optimal configuration and the significant performance gain of coherent case with respect to the non-coherent case and the existing algorithms with single relay selections.

\end{abstract}

\begin{IEEEkeywords}
Cohabitation of radar and communication, interference suppression, full-duplex cooperative communications, resource allocation, relay selection scheme, dynamic spectrum sharing.
\end{IEEEkeywords}
\IEEEpeerreviewmaketitle

\section{Introduction}

Cohabitation of military radar and commercial communications is one of the most promising technologies for 5-Generation (5G) networks, attracting significant attention from researchers \cite{Shatov2024, Doly2024}. 
Moreover, cognitive radio is the innovative technology for addressing today's spectrum shortage due to the flawless multimedia applications supported by the 5G networks.
It is challenging that the commercial 5G networks utilize the common spectrum bands for opportunistic data transfer with the military radar network but not cause any possible spectrum risk for both.
In the following discussion, we will interchangeably refer to the commercial 5G network/user as the secondary network/user or the low-priority network/user and the radar network/user as the primary network/user or the high-priority network/user.
One potential solution is using the dynamic spectrum sharing paradigm, where secondary users (SUs) employ the spectrum sensing technique to search for spectrum holes and opportunistically exploit them. 
Note that primary users (PUs) have higher priority than SUs in transmitting data over the underlying spectrum.
Instead of using spectrum sensing, hierarchical spectrum sharing between primary networks and secondary networks (i.e. an underlay cognitive radio network) is also received much attention from researchers.
Furthermore, the relays operate in the full-duplex (FD) mode, i.e. they can simultaneously transmit and receive data on the same spectrum band \cite{Tan17a, Tan15, Tan16, Tan19, Tan18}.
Hence, the system utilizing FD relays achieves higher throughput and lower latency that using the half-duplex (HD) relays.
So, this paper considers full-duplex underlay cognitive radio network (FDUCRN), where radar and commercial 5G networks transmit concurrently over the same spectrum, under the interference constraint at PU receivers \cite{Tan17a}.

There is a unique challenge for design, implementation and analysis of FDUCRNs, comparing to those of HDUCRNs because the power leaks from the transmitter to the receiver at the FD transceiver and hence causes the presence of self-interference.
Furthermore, it is also challenging to mitigate interference at the high-priority military receivers caused by the data transmission of the commercial 5G network. 
Since the military operations require the advanced, secure and reliable communication network that can support the comprehensive Quality of Service implementation and achieve the graceful degradation.
To address these critical issues, we develop the dynamic interference suppression mechanisms for radar and communication cohabitation design.
In particular, we propose the joint dynamic power allocation and multi-relay selection for FDUCRNs under the self-interference consideration, where low-priority relays employ the amplify-and-forward (AF) protocol.
This work is different from our previous work \cite{Tan17a}, where we consider the more complicated case of multi-relay selection and control the powers allocated to the low-priority users.
We assume that transmit phase information and channel state information are available in the coherent scenario, while only channel state information is available in the non-coherent case.

The contributions of this paper can be summarized as follows.
1) We firstly model the cohabitation of radar and 5G communication network. 
Then, we investigate the non-coherent scenario and formulate the problem of joint power control and relay selection as well as analyze the system rate of low-priority 5G network. 
In fact, we optimize the transmission rate of the low-priority 5G network with the optimal configuration of parameters, including of the transmit powers at the low-priority source and low-priority relays, while considering the power constraints at the low-priority source and low-priority relay as well as the interference constraint at the radar receiver. 
We propose the algorithms including single and multi-update algorithms to solve the optimization problem. 
2) In the coherent scenario, we control the transmit power and perform phase adjustment at FD relays, where a phase of a transmitted signal is adjusted at each relay so that we can suppress the interference at the high-priority receiver. 
Moreover, we develop the alternative optimization mechanism to solve the non-convex problem of power allocation, phase regulation and relay selection. 
3) Extensive numerical results are presented to illustrate the impacts of different parameters on the transmission rate as well as the significant performance gain of FDUCRNs with coherent mechanism compared to FDUCRNs with non-coherent mechanism  and the traditional HDUCRNs. 
Also, the multi-update algorithm achieves rates close to the optimal solutions but has the advantage to reduce significantly time consuming compared to the single-update algorithm and exhausted algorithm.

The remaining of this paper is organized as follows. We introduce the cohabitation model of military radar and commercial 5G communication in Section~\ref{SystemModel}. Then, we propose the joint power allocation and multi-relay selection in Section~\ref{PAFRS_Problem}.
We investigate the non-coherent case in Section~\ref{PCRS_Configuration_NonCoh}, while we consider the coherent scenario in Sesion~\ref{PCRS_Configuration_Coh}.
Next, we present our numerical results in Section~\ref{Results}.
Finally, concluding remarks and future directions are highlighted in Section~\ref{conclusion}.

\section{Cohabitation Model of Military Radar and Commercial Communication}
\label{SystemModel}

\subsection{Scenario Consideration of Radar and 5G Cohabitation}
\begin{figure}[!t]
\centering
\includegraphics[width=75mm]{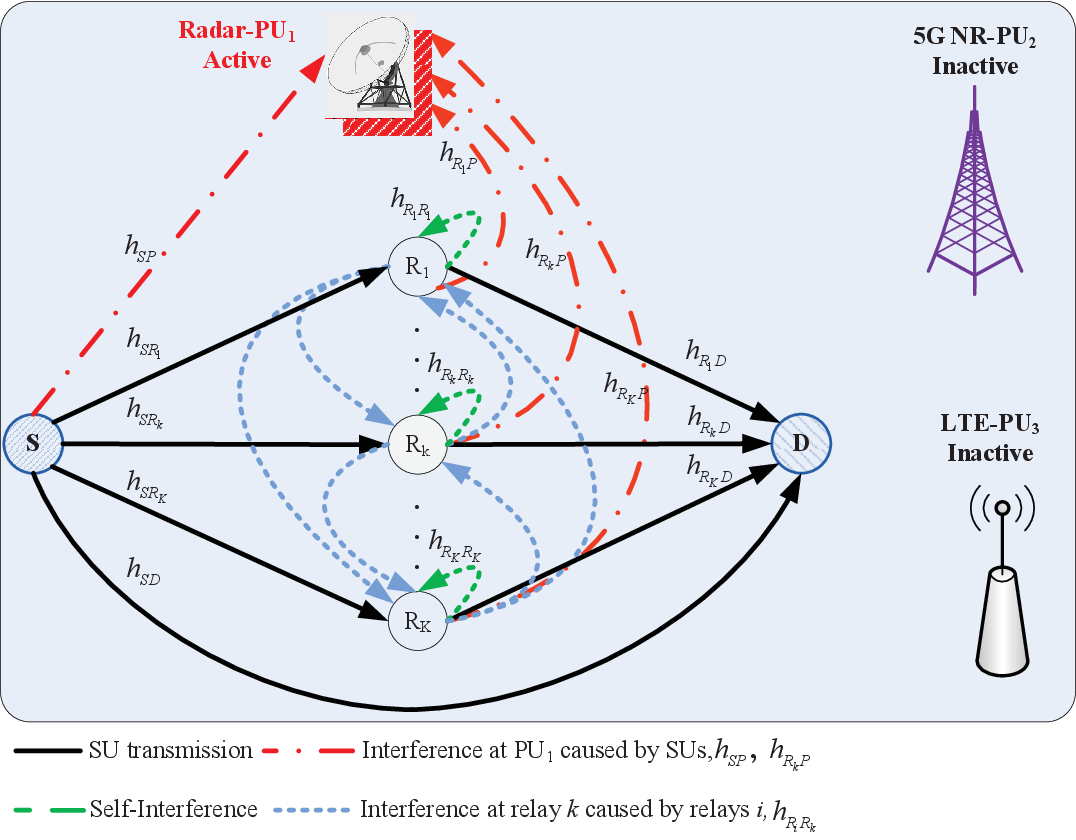} 
\caption{System model of the cognitive full-duplex relay network with three kinds of PUs: Radar, 5G NR and LTE.}
\label{MCRN}
\vspace{-0.5cm}
\end{figure}

We consider the scenario of coexistence of military and commercial wireless spectrum, where we utilize the sensing outcomes for resource allocation (see Fig.~\ref{MCRN} for scenario considered). 
Here, the low-priority network will adaptively allocate the transmit power to its users/UEs based on the outcome of the spectrum sensing.
For example, we can allocate the maximum transmit power to UEs if there is no signal observed on the spectrum band (i.e. noise only).
If there is a radar signal appearing on the frequency band, we will perform the regulation of coherent phases, power allocation and multi-relay selection for the 5G commercial networks.
We will perform the signal processing at the high-priority receiver by considering the coherent scenario, see Fig.~\ref{SRkCRSN}.
In particular, power allocation and multi-relay selection aim to maximize the 5G commercial network performance, whilst the regulation of coherent phases would perform interference suppression for the radar system, i.e. avoiding the interruption of radar communication.
We jointly optimize the set of selected relays and the transmit powers for both the low-priority source and low-priority relays, under the consideration of interference constraint at the high-priority receiver.
In the subsequent sections, we will study the case that both radar and 5G operating on the same channel and will propose the joint power allocation and multi-relay selection for optimizing the system performance.

\subsection{System Model}

In Fig.~\ref{MCRN}, we consider the military radar network with one PU receiver $P$ and a 5G commercial cognitive network with one SU source $S$, $K$ SU relays $R_k$ (where $k = 1, \ldots, K$), and one SU destination $D$.
The SU relays have a capability to operate in the FD mode because they are equipped with FD transceivers.
However, we need to address the self-interference at the SU relays, which is caused by power leakage from the transmitter to the receiver at the transceiver.
For simplicity, we assume that all the other commercial 5G/LTE networks are inactive.
By doing so, we aim to study how to design the cohabitation between the FDUCRN and primary radar network.

We employ the the AF protocol for each SU relay $R_k$, where it amplifies its received signal from the SU source $S$ with a gain of $G_k$ and then forwards this amplified signal to the SU destination, $D$.
Note that this work is different from our previous work \cite{Tan17a}, where multiple relays are selected to help the SU transmission in this work. 
Hence, the performance of each SU relay is affected by not only its self-interference but also the interference from the other SU relays' transmission.

Let $h_{SR_k}$, $h_{R_kD}$, $h_{SD}$, $h_{R_kP}$, $h_{R_jR_k}$ and $h_{R_kR_k}$ denote the channel coefficients of the links $S \rightarrow R_k$, $R_k \rightarrow D$, $S \rightarrow D$, $R_k \rightarrow P$, $R_j \rightarrow R_k$ and $R_k \rightarrow R_k$, respectively.
We denote $h_{SR_k}$, $h_{R_kD}$, $h_{SD}$, $h_{R_kP}$, $h_{R_jR_k}$ and $h_{R_kR_k}$ as the channel coefficients of the links $S \rightarrow R_k$, $R_k \rightarrow D$, $S \rightarrow D$, $R_k \rightarrow P$, $R_j \rightarrow R_k$ and $R_k \rightarrow R_k$, respectively.
We denote $P_S$ by the transmit power of SU source $S$.
We also denote by $x_S(t)$, $y_{R_k}(t)$ and $y_D(t)$ the SU source's signal, the SU relay's transmitted signals and the SU destination's received signals, respectively.

Now, we focus on one specific low-priority relay $R_k$ and perform the signal processing and analysis, as shown in Fig.~\ref{SRkCRSN}.
Here, $y_D(t)$ and $y_1^k(t)$ are defined as the received signals of the SU relay $R_k$ and SU destination $D$, while $y_2^k(t)$ is denoted by the received signal of the SU relay $R_k$ after the amplification.
Let $\mathcal{S}$ be the set of relays, $\mathcal{S} = \left\{1, \ldots, K\right\}$.
We furthermore denote $z_{R_k}(t)$ and $z_D (t)$ by the additive white Gaussian noises (AWGN) with zero mean and variances of $\sigma_{R_k}^2$ and  $\sigma_D^2$, respectively.
We ignore the direct transmission from the SU source to the SU destination due to the large attenuation of this direct transmission link (see \cite{Tan17a} and references therein).
So, the received signals of the SU relay $R_k$ and the SU destination $D$ at time $t$ are expressed as
\beqn
y_1^k(t) \!=\! h_{SR_k} \sqrt{P_S} x_S(t) \!+\! h_{R_kR_k} \left(y_2^k(t) \!+ \! \Delta y^k(t)\right) \nonumber \\
+ \sum_{j \in \mathcal{S} \backslash k} h_{R_jR_k} \left(y_2^j (t) + \Delta y^j(t)\right)  \!+\! z_{R_k}(t), \label{EQN_receiveRD1}\\
y_D(t) \!=\! h_{R_kD} y_{R_k}(t) + h_{SD} \sqrt{P_S} x_S(t) + z_D (t). \label{EQN_receiveRD2} \hspace{0.2cm}
\eeqn

Now, we need to determine the transmitted signals at SU relay $R_k$, $y_{R_k}(t) = y_2^k(t) + \Delta y^k(t).$  
Here, $\Delta y^k(t)$ is the noise and follows the i.i.d. Gaussian distribution with zero mean and variance of $P_{\Delta^k} = \zeta_k P_{R_k}$ \cite{Tan17a}.
Moreover, $y_2^k(t)$ is defined as $y_2^k(t) = f\left(\hat{y}_1^k\right) = G_k \hat{y}_1^k (t-\Delta^k).$ 
Firstly, the SU relay amplifies the signal with the gain, $G_k = \left[P_S \left|h_{SR_k}\right|^2 + \sum_{j \in \mathcal{S}} P_{\Delta^k} \left|h_{R_jR_k}\right|^2 + \sigma_{R_k}^2\right]^{-1/2}$, where $\left|a\right|$ is the amplitude of $a$.
Then, the SU relay delays $\Delta^k$ before transmitting the signal, which is fixed in the noncoherent case and is optimized to minimize the interference at the radar in the coherence case.
After performing interference cancellation, the received signal $\hat{y}_1^k(t)$ can be calculated as  
\beqn
\hat{y}_1^k(t) \!= \sqrt{P_{R_k}} \left(y_1^k(t) - \sum_{j \in \mathcal{S}} h_{R_jR_k} y_2^j(t)\right) \hspace{1.7cm} \nonumber\\
= \sqrt{P_{R_k}} \! \left[\! h_{SR_k} \! \sqrt{P_S} x_S(t) \!+\! \sum_{j \in \mathcal{S}} h_{R_jR_k} \Delta y^j(t) \!+\! z_{R_k}(t) \!\right].\!\!\!\!\!
\eeqn
Finally, the transmitted signals at SU relay $R_k$, $y_{R_k}(t)$ can be expressed as
\beqn
y_{R_k}(t) = G_k h_{SR_k} \sqrt{P_{R_k}} \sqrt{P_S} x_S(t-\Delta^k) + \Delta y^k(t) \hspace{0.7cm} \nonumber\\
\!+ G_k \!\sqrt{P_{R_k}} \sum_{j \in \mathcal{S}} h_{R_jR_k} \Delta y^j(t\!-\!\Delta^k) \!+\! G_k \!\sqrt{P_{R_k}} z_{R_k}(t\!-\!\Delta^k). \label{EQN_y_Rk_trans}
\eeqn
In the above manipulation, we assume that the channels $\left\{h_{R_jR_k}\right\}$ are perfectly estimated.
We utilize the known signal $\left\{y_2^k(t)\right\}$ at SU relay $R_k$ for performing the interference cancellation, however, we still have the residual interference $h_{R_jR_k} \Delta y^j(t)$ in $y_{R_k}(t)$.

\begin{figure}[!t]
\centering
\includegraphics[width=85mm]{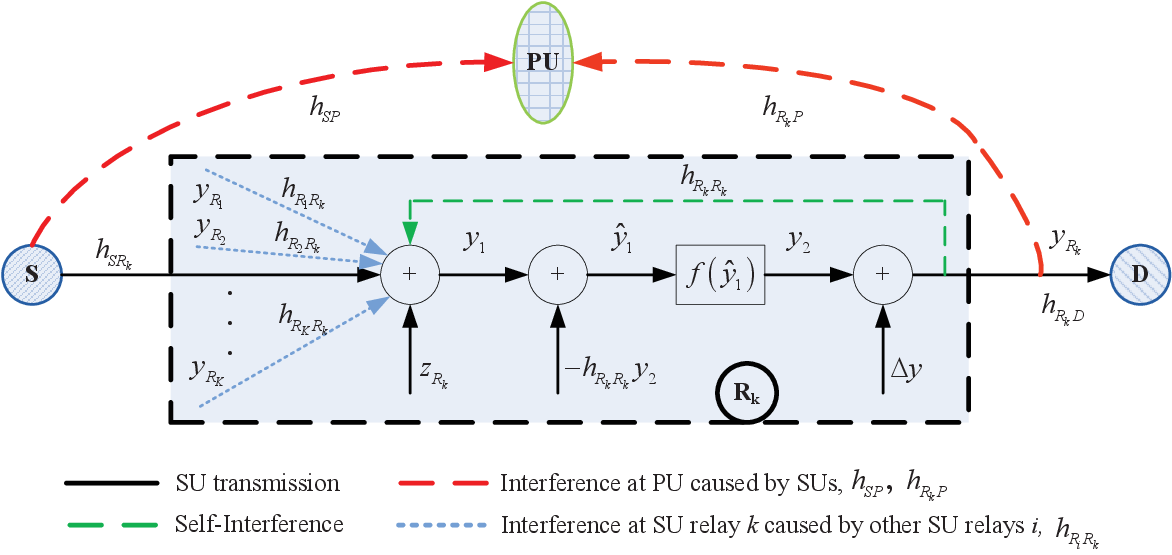}
\caption{The process at FD relay $k$.}
\label{SRkCRSN}
\vspace{-0.5cm}
\end{figure}

\section{Joint Power Allocation and Multi-Relay Selection}
\label{PAFRS_Problem}

In this section, we formulate the system rate maximization for the 5G commercial network, while avoiding the interference to the military radar system.

\subsection{Problem Formulation}
\label{RateOpt}

We denote $\mathcal{C}(P_S, \left\{P_{R_k}\right\})$ by the 5G system rate, which is the function of transmit powers of SU source $S$ and of SU relays $\left\{R_k\right\}$.
In order to protect the PU, the interference caused by the data transmission of SU source and relays, $\mathcal{I}_k$ is required to be at most the threshold of $\overline{\mathcal{I}}_{P}$.
We denote $P_S^{\sf max}$ and $P_{R_k}^{\sf max}$ by the maximum powers of the SU source and SU relay, respectively.
We also denote $\mathcal{S}$ by the group of relays.
Now, the problem of system rate maximization for the FDUCRN can be formulated as follows:
%
\begin{equation}
\label{EQN_OPTRS}
\begin{array}{l}
 {\mathop {\max }\limits_{P_S, \left\{P_{R_k}\right\}}} \quad \mathcal{C}(P_S, \left\{P_{R_k}\right\})  \\ 
 \mbox{s.t.}\,\,\,\, \mathcal{I}\left(P_S, \left\{P_{R_k}\right\}\right) \leq \mathcal{\overline I}_{P}, \\
 \quad \quad 0 \leq P_S \leq P_S^{\sf max},\\
 \quad \quad 0 \leq P_{R_k} \leq P_{R_k}^{\sf max}, k \in \mathcal{S}\\
 \end{array}\!\!
\end{equation}
The first constraint is the requirement of interference at the radar system caused by the 5G transmission $\mathcal{I}\left(P_S, \left\{P_{R_k}\right\}\right)$ is limited.
This constraint would ensure the co-existence of military radar network (PU network) and commercial communication network (SU network).
The second and third constraints on $P_S$ and $\left\{P_{R_k}\right\}$ are for the range of transmit powers of the source and relays. 
Moreover, we need to regulate the SU relays' transmit powers $\left\{P_{R_k}\right\}$ such that the 5G network can balance the trade-off between the higher system rate and better self-interference mitigation.
In the following, we present the derivation of the 5G system rate, $\mathcal{C}(P_S, \left\{P_{R_k}\right\})$ and the interference imposed by SU transmissions, $\mathcal{I}\left(P_S, \left\{P_{R_k}\right\}\right)$.

\subsection{Achievable Rate}
\label{Rate_Formu}

From (\ref{EQN_receiveRD1}) and (\ref{EQN_receiveRD2}), we can derive the 5G network rate for the transmission link $S \rightarrow \left\{R_k\right\} \rightarrow D$ as
\beqn
\mathcal{C} = \log_2 \left[1+ \sum_{k \in \mathcal{S}} \frac{\frac{P_{R_k} \left|h_{R_kD}\right|^2}{\sigma_D^2} \frac{P_S \left|h_{SR_k}\right|^2}{\sum_{j \in \mathcal{S}} \hat{\zeta}_j P_{R_j}+\sigma_{R_k}^2}}  {\mathcal{A}}\right],
\eeqn
where
\beqn
\mathcal{A} &=& 1+ \frac{P_{R_k} \left|h_{R_kD}\right|^2}{\sigma_D^2} + \frac{P_S \left|h_{SR_k}\right|^2}{\sum_{j \in \mathcal{S}} \hat{\zeta}_j P_{R_j}+\sigma_{R_k}^2} \\
\hat{\zeta}_j &=& \left|h_{R_jR_k}\right|^2 \zeta_j. \label{EQN_zeta_hat}
\eeqn
Recall that $\left|a\right|$ is the amplitude of $a$ and we assume the  signal at the destination received directly from the source is negligible.

\subsection{Interference Imposed at PU}
\label{Interference_Formu}

We can evaluate the interference at the PU caused by the 5G network.
This includes the signals received from both the SU source $S$ and the selected relays $\left\{R_k\right\},$ given as 
\beqn
\label{EQN_Inter}
y_I^{\sf PU} (t) \!=\! h_{SP} \sqrt{P_S} x_S(t) \!+\! \sum_{k \in \mathcal{S}} h_{R_kP} y_{R_k}(t)\! +\! z_P(t),
\eeqn
where $y_{R_k}(t)$ is from (\ref{EQN_y_Rk_trans}) and $z_P(t)$ is AWGN with zero mean and variance $\sigma_P^2$.

In the following, we perform interference analysis for both non-coherent and coherent cases. 
In particular, we only consider two scenarios of non-coherent and coherent transmissions for both the links of SU source$\rightarrow$PU receiver and SU relays$\rightarrow$PU receiver, while all other links are set as non-coherent transmissions.  
Furthermore, the phase information in the coherent mechanism is assumed to be known. 
Note that we can utilize implicit feedback and explicit feedback or the channel estimation to obtain this information \cite{Arslan07}.

\subsubsection{Non-coherent Scenario}

We can utilize (\ref{EQN_Inter}) and (\ref{EQN_y_Rk_trans}) to derive the interference received at the PU, given as
\beqn
\mathcal{I}^{\sf non}\left(P_S, \left\{P_{R_k}\right\}\right) = \left|h_{SP}\right|^2 P_S + \sum_{k \in \mathcal{S}} \left|h_{R_kP}\right|^2 \zeta_k P_{R_k} + \nonumber \\
\sum_{k \in \mathcal{S}} \!\!G_k^2 \left|h_{R_kP}\right|^2 \! P_{R_k} \!\!\!\left[\! \left|\!h_{SR_k}\right|^2 \! P_S \!+\! \sum_{j \in \mathcal{S}} \!\!\left|h_{R_jR_k}\right|^2 \! \zeta_j P_{R_j} \!+\! \sigma_{R_k}^2\!\right]\!\!\!
\eeqn
Using some manipulations, we have
\beqn
\mathcal{I}^{\sf non}_k\left(P_S, P_{R_k}\right) = \left|h_{SP}\right|^2 P_S \!+ \!\sum_{k \in \mathcal{S}} \left|h_{R_kP}\right|^2 P_{R_k} \left(1+\zeta_k\right)\!.\!
\eeqn

\subsubsection{Coherent Scenario}

Combining (\ref{EQN_Inter}) and (\ref{EQN_y_Rk_trans}),  we can get the received interference at the PU as
\beqn
\label{EQN_I_k_coh_ori}
\mathcal{\bar{I}}^{\sf coh}\left(P_S, \left\{P_{R_k}\right\}, \left\{\phi_k\right\}\right) = \left|A + \sum_{k \in \mathcal{S}} B_k e^{-j\phi_k}\right|^2,
\eeqn
where $A = h_{SP} \sqrt{P_S} + \sum_{k \in \mathcal{S}} h_{R_kP}\sqrt{\zeta_k P_{R_k}} = \left|A\right| \angle{\phi_A}$, 
\beqn
B_k \!= \!\left(h_{SR_k} \sqrt{P_S} + \sum_{i \in  \mathcal{S}} h_{R_iR_k} \sqrt{\zeta_i P_{R_i}} + \frac{\sigma_{R_i}}{\sqrt{2}} (1+j)\right) \nonumber\\
\times G_k h_{R_kP} \sqrt{P_{R_k}} = \left|B_k\right| \angle{\phi_{B_k}},
\eeqn
$\left|a\right|$ is the amplitude of $a$ and $\phi_k = 2 \pi f_s \Delta^k$, $f_s$ is the sampling frequency.

We further define the following quantities to derive $\mathcal{\bar{I}}^{\sf coh}\left(P_S, \left\{P_{R_k}\right\}, \left\{\phi_k\right\}\right)$.
Let us define $\tilde{B}_k = B_k e^{-j\phi_k} = \left|B_k\right| \angle \phi_{\tilde{B}_k}$ where $\left|\tilde{B}_k\right| = \left|B_k\right|$ and $\phi_{\tilde{B}_k} = \phi_{B_k} - \phi_k$ for $k \in \mathcal{S}$.
We also define $\tilde{B}_0 = \left|\tilde{B}_0\right| \angle \phi_{\tilde{B}_0} = \left|A\right| \angle \phi_A$ (i.e. $\left|\tilde{B}_0\right| = \left|B_0\right| = \left|A\right|$, and $\phi_{\tilde{B}_0} = \phi_A$) and set $\mathcal{\tilde{S}} = \mathcal{S} \cup \left\{0\right\}$.
Let denote the subscripts $R$ and $I$ at $a^R$ and $a^I$ by the real and image of complex number $a$, respectively.
Then, we have 
\beqn
\left|\tilde{B}_k\right|^2 &=& \left(\tilde{B}_k^R\right)^2 + \left(\tilde{B}_k^I\right)^2, \label{EQN_B_abs}\\
\tilde{B}_k^R &=& \left|\tilde{B}_k\right| {\sf Cos}\left(\phi_{\tilde{B}_k}\right) =  \left|B_k\right| {\sf Cos}\left(\phi_{\tilde{B}_k}\right), \label{EQN_B_Real}\\
\tilde{B}_k^I &=& \left|\tilde{B}_k\right| {\sf Sin}\left(\phi_{\tilde{B}_k}\right) = \left|B_k\right| {\sf Sin}\left(\phi_{\tilde{B}_k}\right). \label{EQN_B_Image}
\eeqn

From (\ref{EQN_I_k_coh_ori}), we use some manipulations and rewrite the $\mathcal{\bar{I}}^{\sf coh}\left(P_S, \left\{P_{R_k}\right\}, \left\{\phi_k\right\}\right)$ as 
\beqn \label{I_coh_explain}
\mathcal{\bar{I}}^{\sf coh} = \left|\sum_{k \in \mathcal{\tilde{S}}} \tilde{B}_k^R + j\tilde{B}_k^I\right|^2 
=  \sum_{k \in \mathcal{\tilde{S}}} \left|\tilde{B}_k\right|^2 +\nonumber \\
2 \sum_{k \in \mathcal{\tilde{S}}} \sum_{i \in \mathcal{\tilde{S}}} \left|\tilde{B}_k\right| \left|\tilde{B}_i\right| 
 \times {\sf Cos} \left(\phi_{B_k}-\phi_{B_i}+\phi_i-\phi_k\right).
\eeqn
The detailed derivations are presented in \cite{Techreport}.

We need to perform the pre-processing $\mathcal{\bar{I}}^{\sf coh}_k\left(P_S, \left\{P_{R_k}\right\}, \left\{\phi_k\right\}\right)$ before using it in the optimization problem (\ref{EQN_OPTRS}).
Given $\left(P_S, \left\{P_{R_k}\right\}\right)$, the function $\mathcal{\bar{I}}^{\sf coh}_k\left(P_S, \left\{P_{R_k}\right\}, \left\{\phi_k\right\}\right)$ is minimized over the variable $\left\{\phi_k\right\}$.
This formulated minimization problem is presented as
\vspace{0.0cm}
\begin{equation}
\label{EQN_OPT_PHI}
 {\mathop {\min }\limits_{\left\{\phi_k\right\}} \quad \mathcal{\bar{I}}^{\sf coh}\left(P_S, \left\{P_{R_k}\right\}, \left\{\phi_k\right\}\right) }. 
\end{equation}

\begin{lemma} \label{Lemma_coherent_phase}
Problem (\ref{EQN_OPT_PHI}) is a nonconvex optimization problem for variables $\left\{\phi_k\right\}$. 
\end{lemma} 

\begin{proof} The proof is provided in Appendix~\ref{Lemma_coherent_phase_proof}. \end{proof}

Since, Problem (\ref{EQN_OPT_PHI}) is not convex according to Lemma~\ref{Lemma_coherent_phase}, it is extremely hard to determine the optimal solutions
Based on Lemma~\ref{Lemma_coherent_phase}, it is very hard to determine the optimal solutions for Problem (\ref{EQN_OPT_PHI}).
Hence, we convert the original problem to the alternative problem as follows.
In particular, we propose the in-phase and anti-phase regulation where each relay $k$ regulates the coherent phase $\phi_k$, such that the $\phi_{\tilde{B}_k} = \phi_{\tilde{B}_0} = \phi_A$ (in-phase case) or $\phi_{\tilde{B}_k} = \pi - \phi_{\tilde{B}_0} = \pi - \phi_A$ (anti-phase case).
Let us define $\mathcal{S}_{IP}$ and $\mathcal{S}_{AP}$ as the subset of relays belonging to in-phase cases and the subset of relays belonging to anti-phase cases, respectively.
From (\ref{I_coh_explain}), we can rewrite $\mathcal{\bar{I}}^{\sf coh}\left(P_S, \left\{P_{R_k}\right\}, \left\{\phi_k\right\}\right)$ as
\beqn
\mathcal{\bar{I}}^{\sf coh}\left(P_S, \left\{P_{R_k}\right\}, \left\{\phi_k\right\}\right) = \left(\sum_{k \in \mathcal{S}_{IP} \cup \left\{0\right\}} \!\!\!\!\!\left|B_k\right| \!- \!\!\!\sum_{l \in \mathcal{S}_{AP}} \!\!\left|B_l\right|\right)^2.
\eeqn
Here, we can see that the $\mathcal{\bar{I}}^{\sf coh}$ is now the function of sets $\mathcal{S}_{IP}$ and $\mathcal{S}_{AP}$, i.e. $\mathcal{\bar{I}}^{\sf coh}\left(P_S, \left\{P_{R_k}\right\}, \mathcal{S}_{IP}, \mathcal{S}_{AP}\right)$.
So, the original problem can be equivalently transformed to the alternative problem as
\begin{equation}
\label{EQN_OPT_PHI_ALT}
 {\mathop {\min }\limits_{\mathcal{S}_{IP}, \mathcal{S}_{AP}} \quad \mathcal{\bar{I}}^{\sf coh}\left(P_S, \left\{P_{R_k}\right\}, \mathcal{S}_{IP}, \mathcal{S}_{AP}\right) } 
\end{equation}

The problem (\ref{EQN_OPT_PHI_ALT}) is the two-way number partitioning problem \cite{Graham69}.
To solve this problem, we aim to divide the given set of number $\mathcal{\tilde{S}}$ into two subsets $\mathcal{S}_{IP}\cup \left\{0\right\}$ and $\mathcal{S}_{AP}$ such that the sum of numbers in $\mathcal{S}_{IP}\cup \left\{0\right\}$ is as nearly equal to that in $\mathcal{S}_{AP}$ as possible.
It implies that we will determine $\mathcal{S}_{IP}\cup \left\{0\right\}$ and $\mathcal{S}_{AP}$ such that the difference sum of numbers in these subsets, $\mathcal{D}$ is as close to zero as possible.
We develop the complete greedy algorithm to solve this problem and the summary of this algorithm and its explanation are given as follows.
Other algorithms can be found in \cite{Graham69}.

\begin{algorithm}[ht]
\caption{\textsc{Number Partitioning Algorithm}}
\label{OPT_NUM_PARTITIONING}
\begin{algorithmic}[1]

\STATE Initialize $\mathcal{S}^0_1 = \emptyset$, $\mathcal{S}^0_2 = \emptyset$.

\STATE Calculate $B_k$, $k \in \mathcal{\tilde{S}}$.

\STATE Sort $B_k$ in descending order.

\STATE Create and search a binary tree.

\STATE Stop and return output $\left(\mathcal{S}_1,\mathcal{S}_2\right)$ when the prune condition is satisfied or we reach the last leaf. 
Otherwise, return step 4.

\STATE If $0 \in \mathcal{S}_1$, then $\left(\mathcal{S}_{IP},\mathcal{S}_AP\right) = \left(\mathcal{S}_1 \backslash \left\{0\right\},\mathcal{S}_2\right)$. Otherwise, (i.e., $0 \in \mathcal{S}_2$), $\left(\mathcal{S}_{IP},\mathcal{S}_AP\right) = \left(\mathcal{S}_2 \backslash \left\{0\right\},\mathcal{S}_1\right)$.

\end{algorithmic}
\end{algorithm}
At step 4, each level is corresponding to a different number assignment.
Moreover, each branch assign that index of number to one subset or the other. 
The prune conditions are as follows:

\noindent 1) \textbf{Condition 1}: If we find the leaf $m$ (or complete partition) with $\mathcal{D}_m < \epsilon$ ($\epsilon$ is predetermined small number) then stop running.

\noindent 2) \textbf{Condition 2:} At branch point $m$, if $\mathcal{D}_m - \sum_{k \in \mathcal{S} \backslash \mathcal{S}_1^m \backslash \mathcal{S}_2^m} B_k < \epsilon$ ($\mathcal{S}_1^m$ and $\mathcal{S}_2^m$ are the index sets of all $B_k$ that are already assigned at branch point $m$), then stop and assign all the remaining $k$ ($k \in \mathcal{\tilde{S}} \backslash \mathcal{S}_1^m \backslash \mathcal{S}_2^m$) to the subset with lower sum.

\noindent 3) \textbf{Condition 3:} At branch point $m$, if $\mathcal{D}_m = 0$, the next $k$ (index of number $B_k$) can be assigned to one of $\mathcal{S}_1^m$ and $\mathcal{S}_2^m$.

We furthermore clarify step 4 in Alg.~\ref{OPT_NUM_PARTITIONING} to make it clear.
Example for the tree creation is in Fig.~\ref{NUM_PART}, where we consider the set of $\mathcal{\tilde{S}} = \left\{B_k\right\} = {4,5,6,7,8}$.
For simplicity, we use here the set for these numbers $B_k$ instead of the indexes $k$.
We use $\mathcal{D}_m\left|\mathcal{\tilde{S}}_m\right.$, where $\mathcal{D}_m$ is the difference of between the sum of set $\mathcal{S}_1^m$ and the sum of set $\mathcal{S}_2^m$ at branch $m$, and $\mathcal{\tilde{S}}_m = \mathcal{\tilde{S}} \backslash \mathcal{S}_1^m \backslash \mathcal{S}_2^m$ is the remaining set of unassigned numbers at branch $m$.
For example, at branch 2, we have $\mathcal{D}_m = 8-7=1$ and the remaining set is $\mathcal{\tilde{S}}_m =\left\{6,5,4\right\}$ after assigning 8 to set 1 $\mathcal{S}_1^2$ and 7 to set 2 $\mathcal{S}_2^2$.
Note that we always assign the last number $B_k$ to the subset with lower sum to reduce the number of leaves (i.e. reduce time consuming).
We can observe that case at branches 6, 7, 8 and 9.
For branch 3, we can see that condition 2 is satisfied, so it stops and gives the output.

\begin{figure}[!t]
\centering
\includegraphics[width=80mm]{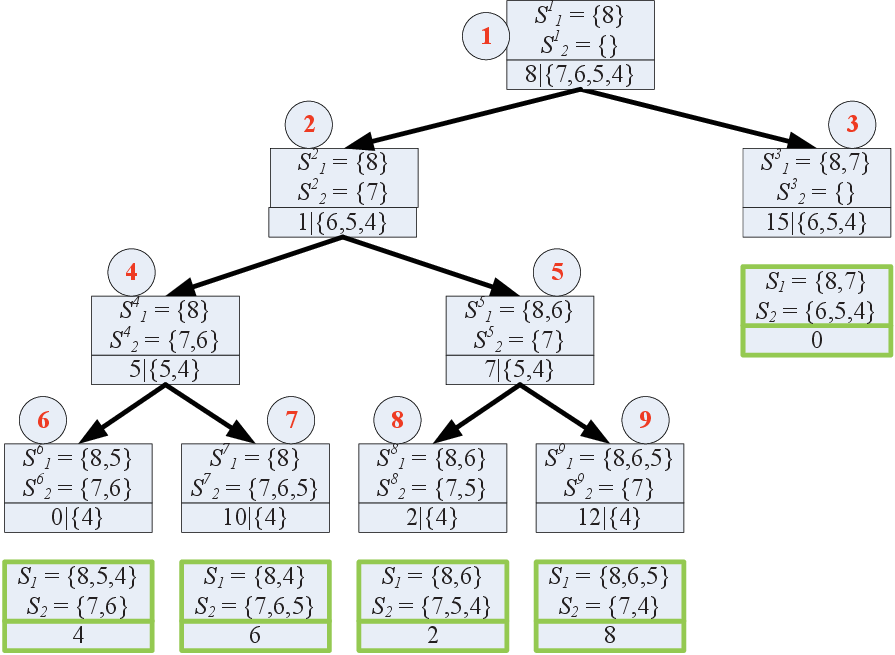}
\caption{The example of tree creation and searching.}
\label{NUM_PART}
\end{figure}

\subsection{Multi-relay Selection}

We generate the combination sets of relays and then perform the resource allocation based on the specific set of relays $\mathcal{S}$.
We aim to get the optimal set $\mathcal{S}^{max}$, which maximizes the achievable rate of FDUCRN. 
This optimal set $\mathcal{S}^{max}$ can be determined by using the assignment mechanisms \cite{Tan19, Tan18}.
Note that for each selected $\mathcal{S}$, we need to perform the power allocation as well as optimize the phase regulation in  the coherent scenario.
It means that we develop the iterative algorithm to obtain the optimal set $\mathcal{S}^{max}$.
In the following, we will derive the algorithms for determining the optimal power allocation (in both coherent and non-coherent cases) and phase regulation (in the coherent case).

\section{Power Allocation in the Non-coherent Scenario}
\label{PCRS_Configuration_NonCoh}

This object function of achievable rate $C$ has the formulation $\log_2(1+x)$, which is strictly increasing with respect to the variable $x$. 
It means that we can convert the original problem (\ref{EQN_OPTRS}) to the new optimization problem, which optimizes the new object $x$ instead of $\log_2(1+x)$.
Also, we assume that $\sum_{j \in \mathcal{S}} \hat{\zeta}_j P_{R_j} >> \sigma_{R_k}^2$, i.e. the self-interference power at the 5G SU relay is much larger than the noise power.
So, the term $\sigma_{R_k}^2$ is omitted in the object function.
Hence, the problem (\ref{EQN_OPTRS}) can be rewritten as
%
\vspace{0.0cm}
\begin{equation}
\label{EQN_OPTRS_1}
\begin{array}{l}
 {\mathop {\max }\limits_{P_S, \left\{P_{R_k}\right\}}} \quad \mathcal{\bar C}(P_S, \left\{P_{R_k}\right\})  \\ 
 \mbox{s.t.}\,\,\,\, \mathcal{I}^{\sf non}\left(P_S, \left\{P_{R_k}\right\}\right) \leq \mathcal{\overline I}_{P}, \\
 \quad \quad 0 \leq P_S \leq P_S^{\sf max}, \\
 \quad \quad 0 \leq P_{R_k} \leq P_{R_k}^{\sf max}, k \in \mathcal{S}, \\
 \end{array}\!\!
\end{equation}
where $\hat{\zeta}_j$ is from (\ref{EQN_zeta_hat}) and $\mathcal{\bar C}(P_S, \left\{P_{R_k}\right\}) = 1/\mathcal{\bar A}\sum_{k \in \mathcal{S}}\frac{P_{R_k} \left|h_{R_kD}\right|^2}{\sigma_D^2} \frac{P_S \left|h_{SR_k}\right|^2}{\sum_{j \in \mathcal{S}} \hat{\zeta}_j P_{R_j}}$, $\mathcal{\bar A} = 1+ \frac{P_{R_k} \left|h_{R_kD}\right|^2}{\sigma_D^2} + \frac{P_S \left|h_{SR_k}\right|^2}{\sum_{j \in \mathcal{S}} \hat{\zeta}_j P_{R_j}}$.

To solve the problem (\ref{EQN_OPTRS_1}), we need to characterize the optimization problem by introducing Lemmas~\ref{Lemma_noncoh1} and \ref{lem: nonco-convex} as follows.
\begin{lemma} \label{Lemma_noncoh1}
Problem (\ref{EQN_OPTRS_1}) is not a convex optimization problem for variables $\left(P_S, \left\{P_{R_k}\right\}\right)$.
\end{lemma} 

\begin{proof} The proof is provided in Appendix~B in \cite{Techreport}. \end{proof}

\begin{lemma} 
Given $P_S \in \left[0, P_S^{\sf max}\right]$ and $P_{R_j} \in \left[0, P_{R_j}^{\sf max}\right]$ for $j \in \mathcal{S} \backslash k$, problem (\ref{EQN_OPTRS_1}) is a convex optimization problem in terms of $P_{R_k}$. Similarly, given $P_{R_k} \in \left[0, P_{R_k}^{\sf max}\right]$ for $k \in \mathcal{S}$, problem (\ref{EQN_OPTRS_1}) is also a convex optimization problem in terms of $P_S$. \label{lem: nonco-convex}
\end{lemma} 

\begin{proof} The proof is provided in Appendix~D in \cite{Techreport}. \end{proof}

To address the non-convex optimization problem, we now derive the alternative optimization problem to determine the solution for problem (\ref{EQN_OPTRS_1}) according to Lemma \ref{lem: nonco-convex}.  
In fact, when we optimize one variable and fix all the others, the problem turns to be convex.
It means that we obtain the convex optimization problem in each consideration and hence utilize the standard approaches \cite{Boyd04} to solve it.
We also propose the multi-update algorithm to solve the problem, where we can update multiple transmit powers of relays, $P_{R_k}$ at each step.
The simulation results in Section~\ref{Results} show that our new proposal achieves good rate performances, which are close to the optimal solution as well as reduces the CPU time.

\section{Power Allocation and Phase Regulation in the Coherent Scenario}
\label{PCRS_Configuration_Coh}

Assume that $\sum_{j \in \mathcal{S}}\hat{\zeta}_j P_{R_j} >> \sigma_{R_k}^2$, i.e. the noise power is negligible when compared to the  power of the self-interference at the selected relay $R_k$.
Hence, problem (\ref{EQN_OPTRS}) can be converted to the new formulation as
\begin{equation}
\label{EQN_OPTRS_2}
\begin{array}{l}
 {\mathop {\max }\limits_{P_S, \left\{P_{R_k}\right\}}} \quad \mathcal{\bar C}^{\sf coh}(P_S, \left\{P_{R_k}\right\})  \\ 
 \mbox{s.t.}\,\,\,\, \mathcal{\bar{I}}^{\sf coh}\left(P_S, \left\{P_{R_k}\right\}\right) \leq \mathcal{\overline I}_{P}, \\
 \quad \quad 0 \leq P_S \leq P_S^{\sf max}, \\
 \quad \quad 0 \leq P_{R_k} \leq P_{R_k}^{\sf max}, k \in \mathcal{S},\\
 \end{array}\!\!
\end{equation}  
where $\hat{\zeta}_j$ is calculated in (\ref{EQN_zeta_hat}) and
\beqn
\mathcal{\bar C}^{\sf coh}(P_S, \left\{P_{R_k}\right\}) = \sum_{k \in \mathcal{S}} \frac{\frac{P_{R_k} \left|h_{R_kD}\right|^2}{\sigma_D^2} \frac{P_S \left|h_{SR_k}\right|^2}{\sum_{j \in \mathcal{S}} \hat{\zeta}_j P_{R_j}}} {1 + \frac{P_{R_k} \left|h_{R_kD}\right|^2}{\sigma_D^2} + \frac{P_S \left|h_{SR_k}\right|^2}{\sum_{j \in \mathcal{S}} \hat{\zeta}_j P_{R_j}}}.
\eeqn
Next, we introduce two new variables $p_S$ and $p_{R_k}$, with $p_S = \sqrt{P_S}$ and $p_{R_k} = \sqrt{P_{R_k}}$. 
So, we can have the new optimization problem, given as
\begin{equation}
\label{EQN_OPTRS_3}
\begin{array}{l}
 {\mathop {\max }\limits_{p_S, \left\{p_{R_k}\right\}}} \quad \mathcal{\breve{C}}^{\sf coh}(p_S, \left\{p_{R_k}\right\})  \\ 
 \mbox{s.t.}\,\,\,\, \mathcal{\bar{I}}^{\sf coh}\left(p_S, \left\{p_{R_k}\right\}\right) \leq \mathcal{\overline I}_{P}, \\
 \quad \quad 0 \leq p_S \leq \sqrt{P_S^{\sf max}}, \\
 \quad \quad 0 \leq p_{R_k} \leq \sqrt{P_{R_k}^{\sf max}}, k \in \mathcal{S},\\
 \end{array}\!\!
\end{equation} 
where $\mathcal{\breve{C}}^{\sf non}_k(p_S, p_{R_k})$ is the objective function, i.e.
\beqn
\mathcal{\breve{C}}^{\sf non}_k(p_S, p_{R_k}) = \sum_{k \in \mathcal{S}} \frac{\frac{p_{R_k}^2 \left|h_{R_kD}\right|^2}{\sigma_D^2} \frac{p_S^2 \left|h_{SR_k}\right|^2}{\sum_{j \in \mathcal{S}} \hat{\zeta}_j p_{R_j}^2}} {1 + \frac{p_{R_k}^2 \left|h_{R_kD}\right|^2}{\sigma_D^2} + \frac{p_S^2 \left|h_{SR_k}\right|^2}{\sum_{j \in \mathcal{S}} \hat{\zeta}_j p_{R_j}^2}}.
\eeqn
We now perform analysis and characterization for the solution of this optimization problem in Lemmas~\ref{Lemma_coh1} and \ref{lem: co-convex}.

\begin{lemma} \label{Lemma_coh1} Problem (\ref{EQN_OPTRS_3}) is a nonconvex optimization problem for variables $\left(p_S,\left\{p_{R_k}\right\}\right)$. \end{lemma}

\begin{proof} The proof is provided in Appendix~E in \cite{Techreport}. \end{proof}

\begin{lemma} Given $p_S \in \left[0, \sqrt{P_S^{\sf max}}\right]$ and $p_{R_j} \in \left[0, \sqrt{P_{R_j}^{\sf max}}\right]$, $\forall j \in \mathcal{S} \backslash k$, problem (\ref{EQN_OPTRS_3}) is a convex optimization problem for variable $p_{R_k}$. Similarly, given $p_{R_k} \in \left[0, \sqrt{P_{R_k}^{\sf max}}\right]$ , $\forall k \in \mathcal{S}$, problem (\ref{EQN_OPTRS_3}) is also a convex optimization problem for variable $p_S$. \label{lem: co-convex}\end{lemma}

\begin{proof} The proof is provided in Appendix~G in \cite{Techreport}. \end{proof}

According to Lemma \ref{lem: co-convex}, we develop the alternative optimization to determine the solutions of (\ref{EQN_OPTRS_3}).
In particular, we can divide the problem into multiple sub-problems, which are convex optimization problems and their solutions can be determined using standard optimization methods \cite{Boyd04}. 
We employ the multi-update algorithm to solve the problem, where we can update the transmit powers for multiple relays at each step.
Simulation results in Section~\ref{Results} validate and confirm the robustness and benefit gain in CPU-time reduction of this proposed algorithm.

\section{Numerical Results}
\label{Results}

\begin{table} 
\centering
\caption{System rate vs $\mathcal{\bar I}_P$ ($P_{\sf max} = 20 dB$, $\zeta$ = 0.01, Co: Coherent scenario, NCo: Non-coherent scenario)}
\label{table1}
\scriptsize
\begin{tabular}{|c|c|c|c|c|c|c|c|}
\hline 
\multicolumn{2}{|c|}{$\mathcal{\bar I}_P$ (dB)}     
       & 0   & 2   & 4   & 6   & 8 & 10\tabularnewline
\hline 
\hline 
          & Optimal & 7.6741  &  7.6818  &  7.6840 & 7.6847  & 7.6865 &  7.6873  \tabularnewline
\cline{2-8} 
Co      & Greedy 1 & 7.6655 & 7.6721 & 7.6751 & 7.6752 & 7.6752 & 7.6753   \tabularnewline
\cline{2-8} 
          & Greedy 2  & 7.6565 &  7.6595 & 7.6608 & 7.6608 & 7.6610 & 7.6652  \tabularnewline
\cline{2-8} 
          & $\Delta \mathcal{C}_1 (\%)$ & 0.1121  &  0.1263  &  0.1158  &  0.1236  &  0.1470  &  0.1561    \tabularnewline
\cline{2-8} 
          & $\Delta \mathcal{C}_2 (\%)$ & 0.2293  &  0.2903  &  0.3019  &  0.3110  &  0.3318  &  0.2875   \tabularnewline
\hline 
\hline 
        & Optimal & 4.4638 &   4.8985  &  5.3447  &  5.4611  &  5.5358  &  5.7043  \tabularnewline
\cline{2-8} 
NCo  & Greedy 1  & 4.4500  &  4.8902  &  5.2902  &  5.4043  &  5.5269  &  5.7026  \tabularnewline
\cline{2-8} 
        & Greedy 2 & 4.4254  & 4.8312 & 5.2507 & 5.4013 & 5.4967 & 5.6704  \tabularnewline
\cline{2-8} 
          & $\Delta \mathcal{C}_1 (\%)$ & 0.3092  &  0.1694  &  1.0197  &  1.0401  &  0.1608  &  0.0298 \tabularnewline
\cline{2-8} 
          & $\Delta \mathcal{C}_2 (\%)$ & 0.8603  &  1.3739  &  1.7588  &  1.0950  &  0.7063  &  0.5943  \tabularnewline
\hline
\end{tabular}
\end{table}

\begin{figure}[!t]
\centering
\includegraphics[width=70mm]{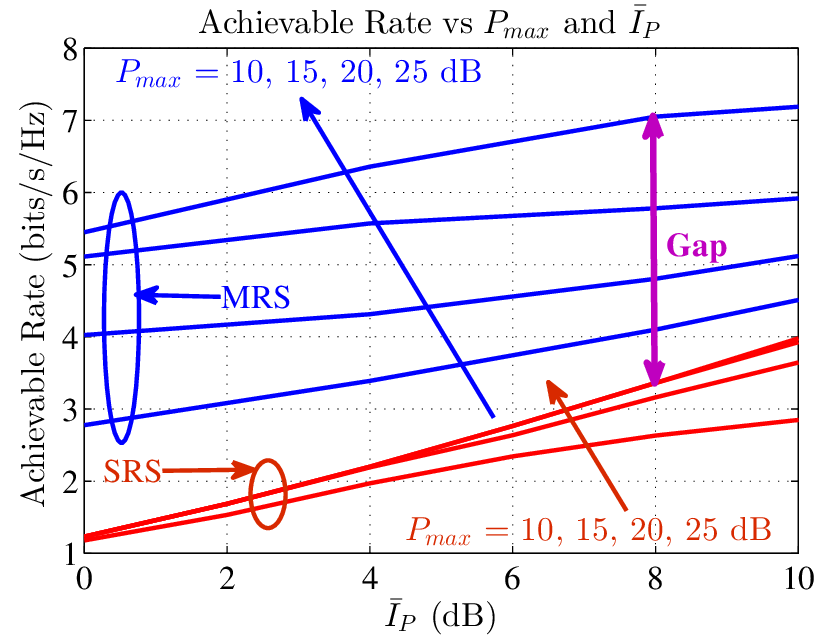}
\caption{System rate versus the interference constraint $\mathcal{\bar I}_P$
 for $K = 4$, $\zeta = 0.001$, $P_{\sf max} = \left\{10, 15, 20, 25\right\}$ dB, non-coherent scenario, and both single- and multi-relay selections.}
\vspace{-10pt}
\label{1_Rate_vs_I_bar_P_P_max_10152025_zeta_0001_MS_Non}
\end{figure}

The parameter setting of the FDUCRN is summarized as follows. 
We set the noise power at each node to one.
We assume that the 5G communication link follows a Rayleigh fading with variance one (i.e. $\sigma_{SR_k}$ = $\sigma_{R_kD}$ = 1), except the $S-D$ link (i.e. $\sigma_{SD}$ = 0.1) and the links of the SU relay $i$-SU relay $j$ (i.e. $\sigma_{R_iR_j} \in \left[0.5, 1\right]$).
We have the similar assumption to the links of the SU relay-radar/PU receiver and SU source-radar/PU receiver with corresponding variances $\left\{\sigma_{SP}, \sigma_{R_kP}\right\} \in \left[0.8, 1\right]$.
We also assume that we have the same impact of imperfect channel estimation for all relays, i.e. $\zeta$.

In the first experiment, we compare our proposed algorithms' system rates to those of the optimal brute-force mechanism so that we can show the efficacy of our developed algorithms.
Table~\ref{table1} show the rate performances of the proposed and optimal algorithms for both coherent and non-coherent scenarios, given the following parameter setting and scenarios, $\mathcal{\bar I}_P = \left\{0, 2, 4, 6, 8, 10\right\}$ dB, $\zeta$ = 0.01, 4 SU relays and $P_{\sf max} = 20$ dB.
In particular, we consider the two algorithms, where we optimize one transmit power variable of one relay for one step in the first algorithm, \textbf{Greedy 1}, and simultaneously optimize two transmit power variables of two relays in the second algorithm, \textbf{Greedy 2}. 
It is easily seen that there are slightly different between the rate performances of our proposed algorithms and the optimal solutions for both coherent and non-coherent scenarios (i.e. the errors are almost always much lower than 1\%).
Moreover, \textbf{Greedy 2} performs good results while it saves the time consuming due to performing the two-variable update.

We also see that the coherent mechanism always outperforms the non-coherent mechanism because we carefully regulate the phase to reduce the interference at the radar system caused by the 5G transmissions.  As a result, this allows to increase the transmit powers for the 5G commercial networks and hence it can achieve the higher system rates. 
As expected, when QSIC at the FD relay increases (i.e. the self-interference increases) then the the achievable rate decreases.
Now, we compare the achievable rates of our proposed FDUCRN with multi-relay selection (MRS) and the FDUCRN with single-relay selection (SRS) \cite{Tan17a} in Fig.~\ref{1_Rate_vs_I_bar_P_P_max_10152025_zeta_0001_MS_Non}. 
Note that only the best relay is selected in the FDUCRN with SRS.
Here, we consider the FDUCRN $4$ SU relays with $\zeta$ = 0.001 for the non-coherent scenario. 
It is easily observed that our proposed FDUCRN with MRS significantly outperforms the FDUCRN with SRS.
More results are presented in \cite{Techreport}.

\section{Conclusions and Future Directions}
\label{conclusion} 

\subsection{Conclusions}
This paper studied cohabitation of military radar and commercial communication networks.
We derived the joint power control and multi-relay selection in FDUCRNs, which is used for the commercial communication network such that it can opportunistically transmit data over the licensed spectrum bands of military radar network. 
In particular, we formulated the rate optimization problem and performed analysis of the 5G system rate under the interference constraints.
We then proposed efficient algorithms to configure the optimal network parameters.
Our design and analysis have considered both the FD communication capability and the self-interference of transceiver at the 5G relays.
Especially, we have investigated the interference suppression at the radar system by using both coherent and non-coherent cases.
Extensive numerical results have been presented to demonstrate the significant gains of our coherent method to the interference suppression at the radar system and the self-interference impacts to the 5G network. 

\subsection{Future Directions}

We highlight the future directions based on our current results as follows:

\noindent \textbf{AI-based Spectrum Sensing for Identification of Radar and 5G/4G Signals.}
We consider the scenario of the coexistence of the federal and commercial wireless spectrum. 
In particular, we consider the main coverage of radar communication, while the 5G network opportunistically accesses the spectrum by using the cooperative underlay cognitive access.
We derive the AI-based sensing platform, which can sense the spectrum bands as well as identify/classify the signals transmitting over the bands (which are the radar signal or 5G signal or noise only). 
The results would be used for the access part, which is implemented in the subsequent section.
We also investigate the more complicated scenario, where there are multiple active high-priority users, including not only radar but also other 5G and LTE.
We extend to develop the cooperative compressive spectrum sensing, which consists of centralized sensing and distributed sensing \cite{Tan2010, Tan2010a, le2014joint}. 
Finally, we consider the classification problem that the fusion center makes the final classification of the channel availability according to the receiving raw sensor data from sensor nodes. 
This can be solved by using machine learning (ML) techniques, such as the unsupervised ML, supervised ML and semi-supervised ML \cite{Tan18, Tan19, Zahin19, Pervej20}. 

\noindent \textbf{AI-based Dynamic Spectrum Sharing And Cohabitation with Identification and Classification of Spectrum Parameter.}
We will develop the intelligence and learning in spectrum sensing and sharing for cognitive 5G networks, which involve spectrum sensors to provide sensing information and distributed agents to provide spectrum availability information as well as intelligent MAC protocol mechanisms for efficient access. It means that we aim for ensuring the effective spectrum utilization and guaranteeing the transmission quality. In addition, we will develop an intelligent MAC protocol \cite{Tan2016RA, le2014joint}, which can determine the strategies of UEs depending on dynamic environments, i.e. switching between the interweave, underlay and overlay paradigms. 
In the case of an escape from interweave paradigm, we determine whether the UEs will switch between the overlay and underlay paradigms depending on the density of high-priority users and the surrounding environment.
Here, we utilize the AI-enabled automatic modulation recognition \cite{Nandi1998} for supporting resource allocation to ensure QoS, whilst eliminating interference to high-priority users. 
Besides using the deep reinforcement learning \cite{Tan2018DRL, Tan2019DRL}, we also exploit the recent advanced machine learning techniques to support for changing from interweave to underlay and/or overlay. 
For instance, we utilize the multiple agents for different channels in muti-channel scenarios and incorporate with the LSTM-based estimation \cite{Pervej20, Wang2020} of frequency hopping patterns from the other low-priority users. 
We also develop an end-to-end convolutional neural network \cite{Zahin19} to automatically learn features directly from simple representations without requiring expert features. 
The novelty of our proposal is not only contributing to the performance enhancement but also providing the adaptive framework, when compared to the existing static MAC protocol designs.
As a result, this will contribute to improve the communication part for the EdgeAI system \cite{Tan2018DRL, Tan2019DRL, Tan18, Tan2024, Tan19, Wang2020, Letaief2022, Mao2024}.

\section*{Acknowledgment}
This work was supported in part by the Commonwealth Cyber Initiative (CCI) Experiential Learning Program and the South Big Data Innovation Hub Partnership Nucleation program.
\vspace{-2pt}
\bibliographystyle{IEEEtran}


\appendices

\section{Proof of Lemma~ \ref{Lemma_coherent_phase}}
\label{Lemma_coherent_phase_proof}
\vspace{5pt}

Firstly, we find the Hessian matrix of $\mathcal{\bar{I}}^{\sf coh}\left(P_S, \left\{P_{R_k}\right\}, \left\{\phi_k\right\}\right)$, which is presented as $\mathcal{H}_I = \left[\mathcal{H}_{ki}\right]$, where $k, i \in \mathcal{S}$.
Here $\mathcal{H}_{ki} = \frac{\partial^2 \mathcal{\bar{I}}^{\sf coh}}{\partial \phi_k \partial \phi_i}$.

According to the Sylvester's criterion \cite{Horn12}, the Hessian matrix $\mathcal{H}_I$  is positive definite iff $\mathcal{H}_{11} > 0$ and the determinant of $\mathcal{H}_I$ ${\sf det} \left(\mathcal{H}_I\right) > 0$. 
However we can prove that this condition is not satisfied as follows (i.e. $\mathcal{H}_{11}$ is not always larger than zero).
We take first-order partial derivative of $\mathcal{\bar{I}}^{\sf coh}\left(P_S, \left\{P_{R_k}\right\}, \left\{\phi_k\right\}\right)$ with respect to $\phi_1$
\beqn
\frac{\partial \mathcal{\bar{I}}^{\sf coh}}{\partial \phi_1} = 2 \sum_{i \in \mathcal{\tilde{S}} \backslash \left\{1\right\}} \left|\tilde{B}_k\right| \left|\tilde{B}_i\right| {\sf Sin} \left(\phi_{B_k}-\phi_{B_i}+\phi_i-\phi_k\right)
\eeqn
Then, $\mathcal{H}_{11}$ which is the second-order partial derivative of $\mathcal{\bar{I}}^{\sf coh}\left(P_S, \left\{P_{R_k}\right\}, \left\{\phi_k\right\}\right)$ with respect to $\phi_1$ can be determined as
\beqn
\frac{\partial^2 \mathcal{\bar{I}}^{\sf coh}}{\partial \phi_1^2} = -2\!\!\! \sum_{i \in \mathcal{\tilde{S}} \backslash \left\{1\right\}}  \!\!\!\left|\tilde{B}_k\right| \left|\tilde{B}_i\right| {\sf Cos} \left(\phi_{B_k}-\phi_{B_i}+\phi_i-\phi_k\right)
\eeqn

We can easily observe that $\mathcal{H}_{11} = \frac{\partial^2 \mathcal{\bar{I}}^{\sf coh}}{\partial \phi_1^2}$ is not always larger than zero.
Hence, the function $\mathcal{\bar{I}}^{\sf coh}\left(P_S, \left\{P_{R_k}\right\}, \left\{\phi_k\right\}\right)$ is not the strictly convex function.
Thus, the problem (\ref{EQN_OPT_PHI}) is not the strictly convex optimization problem.

\end{document}